\documentclass[11pt]{article}
\usepackage{amsmath,amsthm,amssymb,cite,enumerate}
\usepackage[a4paper,left=0.5in,right=1in]{geometry}
\usepackage[usenames]{xcolor}
\usepackage{graphicx}
\usepackage{epsfig}
\usepackage{dcolumn}
\usepackage{bm}

\usepackage{soul} 

\begin{document}

\def \d {{\rm d}}

\def \bm #1 {\mbox{\boldmath{$m_{(#1)}$}}}

\def \bF {\mbox{\boldmath{$F$}}}
\def \bV {\mbox{\boldmath{$V$}}}
\def \bff {\mbox{\boldmath{$f$}}}
\def \bT {\mbox{\boldmath{$T$}}}
\def \bk {\mbox{\boldmath{$k$}}}
\def \bl {\mbox{\boldmath{$\ell$}}}
\def \bn {\mbox{\boldmath{$n$}}}
\def \bbm {\mbox{\boldmath{$m$}}}
\def \tbbm {\mbox{\boldmath{$\bar m$}}}

\def \0 {^{(0)}}
\def \1 {^{(1)}}

\def \a {\alpha}
\def \B {\beta}

\def \T {\bigtriangleup}
\newcommand{\msub}[2]{m^{(#1)}_{#2}}
\newcommand{\msup}[2]{m_{(#1)}^{#2}}

\newcommand{\be}{\begin{equation}}
\newcommand{\ee}{\end{equation}}

\newcommand{\beqn}{\begin{eqnarray}}
\newcommand{\eeqn}{\end{eqnarray}}
\newcommand{\AdS}{anti--de~Sitter }
\newcommand{\AAdS}{\mbox{(anti--)}de~Sitter }
\newcommand{\AAN}{\mbox{(anti--)}Nariai }
\newcommand{\AS}{Aichelburg-Sexl }
\newcommand{\pa}{\partial}
\newcommand{\pp}{{\it pp\,}-}
\newcommand{\ba}{\begin{array}}
\newcommand{\ea}{\end{array}}

\newcommand{\M}[3] {{\stackrel{#1}{M}}_{{#2}{#3}}}
\newcommand{\m}[3] {{\stackrel{\hspace{.3cm}#1}{m}}_{\!{#2}{#3}}\,}

\newcommand{\tr}{\textcolor{red}}
\newcommand{\tb}{\textcolor{blue}}
\newcommand{\tg}{\textcolor{green}}
\newcommand{\tor}{\textcolor{orange}}

\def\a{\alpha}
\def\g{\gamma}
\def\de{\delta}

\def\b{{\kappa_0}}

\def\R{{\cal R}}
\def\F{{\cal F}}
\def\L{{\cal L}}

\def\e{e}
\def\bb{b}


\newtheorem{theorem}{Theorem}[section] 
\newtheorem{cor}[theorem]{Corollary} 
\newtheorem{lemma}[theorem]{Lemma} 
\newtheorem{proposition}[theorem]{Proposition}
\newtheorem{definition}[theorem]{Definition}
\newtheorem{remark}[theorem]{Remark}

\title{Electromagnetic fields with vanishing scalar invariants}

\author{Marcello Ortaggio\thanks{ortaggio@math.cas.cz} \ and Vojt\v ech Pravda\thanks{pravda@math.cas.cz} \\
Institute of Mathematics of the Czech Academy of Sciences
\\ \v Zitn\' a 25, 115 67 Prague 1, Czech Republic}

\maketitle

\abstract{We determine the class of $p$-forms $\bF$ which possess vanishing scalar invariants (VSI) at arbitrary order in a $n$-dimensional spacetime. Namely, we prove that $\bF$ is VSI {\em if and only if} it is of type N, its multiple null direction $\bl$ is ``degenerate Kundt'', and {$\pounds_{\bl}\bF=0$}. The result is theory-independent. Next, we discuss the special case of Maxwell fields, both at the level of test fields and of the full Einstein-Maxwell equations. These describe electromagnetic non-expanding waves propagating in various Kundt spacetimes. We further point out that a subset of these solutions possesses a {\em universal} property, i.e., they also solve (virtually) any generalized (non-linear and with higher derivatives) electrodynamics, possibly also coupled to Einstein's gravity.
}

\vspace{.2cm}
\noindent
PACS 04.50.+h, 04.20.Jb, 04.40.Nr


\tableofcontents

\section{Introduction}
\label{intro}

{\subsection{Background}}

Synge \cite{syngespec} called electromagnetic {\em null} fields those characterized by the vanishing of the two Lorentz invariants, i.e., 
\be
	F_{ab}F^{ab}=0 , \qquad F_{ab} {}^*F^{ab}=0 .
	\label{null}
\ee
From a physical viewpoint, null Maxwell fields characterize electromagnetic plane waves  \cite{syngespec} as well as the asymptotic behaviour of radiative systems (cf. \cite{penrosebook2} and references therein). Fields satisfying~\eqref{null} single out a unique null direction $\bl$ such that the corresponding energy-momentum tensor can be written as $T_{ab}=A\ell_a\ell_b$ \cite{Ruse36} (cf. also, e.g., \cite{syngespec,Stephanibook}). They also possess the unique property that their field strength at any spacetime point can be made as small (or large) as desired in a suitably boosted reference frame \cite{syngespec,Schroedinger35}. {Moreover, null solutions to the sourcefree Maxwell equations can be associated with shearfree congruences of null geodesics (and viceversa) via the Mariot-Robinson theorem \cite{Stephanibook,penrosebook2} and are therefore geometrically privileged.}

There are further reasons that motivate the interest in fields with the property \eqref{null}. On the gravity side, an analog of null electromagnetic fields is given by {metrics} of Riemann type III and N, for which all the zeroth-order scalar  invariants constructed from the Riemann tensor vanish identically (see \cite{Stephanibook} in 4D and \cite{Coleyetal04vsi} in higher dimensions).\footnote{Hereafter, by ``zeroth-order'' invariants of a certain tensor we refer to the {\em algebraic} ones, i.e., those not involving covariant derivatives of the given tensor. Additionally, we will restrict ourselves to {\em polynomial} scalar invariants (cf. Definition~5.1 of \cite{OrtPraPra13rev}).} These have remarkable properties. For example, all type~N Einstein spacetimes are automatically vacuum solutions of quadratic \cite{MalPra11prd} (in particular, Gauss-Bonnet \cite{PraPra08}) and Lovelock gravity \cite{ReaTanBen14}.
Furthermore, it has been known for some time \cite{Guven87,AmaKli89,HorSte90} that certain type N \pp wave solutions in general relativity (in vacuum or with dilaton and form fields) are classical solutions to string theory to all orders in $\sigma$-model perturbation theory -- in fact, they are also solutions to {\em any} gravity theory in which the ``corrections'' to the field equations can be expressed in terms of scalars and tensors  constructed from the field strenghts and their covariant derivatives (see \cite{Guven87,HorSte90,Horowitz90} for details), and are in this sense {\em universal} \cite{Coleyetal08}.\footnote{{To be precise, in the terminology of \cite{Coleyetal08} the universal property refers only to certain Einstein metrics, whereas here we clearly use the term ``universal'' in a broader sense (which applies not only to the metric -- not necessarily Einstein -- but to the full solution, also including possible matter fields).}}
Apart from being of Riemann type N and universal, the metrics considered in \cite{Guven87,AmaKli89,HorSte90} have the {special} property that all the scalar invariants constructed from the Riemann tensor {\em and} its covariant derivatives vanish and thus belong to the VSI class of spacetimes  \cite{Pravdaetal02,Coleyetal04vsi} (cf. Definition~\ref{def_VSI} below).\footnote{In the vacuum case, recent analysis \cite{HerPraPra14,Herviketal15} has extended the result of \cite{Guven87,AmaKli89,HorSte90} in various directions. In particular, it is now clear that the VSI property is neither a sufficient nor a necessary condition for ``universality'' (however, universal spacetimes must be CSI \cite{HerPraPra14}, i.e., with constant scalar invariants).  
Moreover, certain {non-\pp wave spacetimes of Weyl type III, N} \cite{HerPraPra14} and II (or D) \cite{Herviketal15} can also be universal.}

In view of these results for gravity, one may wonder whether {certain} null (or VSI) Maxwell fields possess a similar ``universal'' property and thus solve {also} generalized theories of electrodynamics. In fact, it was already known to Schr{\"o}dinger \cite{Schroedinger35,Schroedinger43} that all null Maxwell fields solve the equations for the electromagnetic field in any non-linear electrodynamics (NLE; cf., e.g., \cite{Plebanski70}) -- this was later extended to the full Einstein-Maxwell equations including the electromagnetic backreaction on the spacetime geometry \cite{Kichenassamy59,KreKic60,Peres61}. 
However, the case of theories more general than NLE (including also derivatives of the field strength in the Lagrangian, cf., e.g., \cite{Bopp40,Podolsky42}) seems not to have been investigated systematically from this viewpoint. 
As a first step in this direction, it is the purpose of the present paper to determine the class of electromagnetic fields for which all the scalar invariants constructed from the field strength {and its derivatives} vanish identically (VSI), which is obviously a subset of null fields. We will also point out a few examples possessing the universal property, while a more detailed study of the latter will be presented elsewhere.

Various extensions of Einstein-Maxwell gravity exist in which the number of spacetime dimensions~$n$ may be greater than four, and the electromagnetic field is represented by a rank-$p$ form (cf., e.g., \cite{Teitelboim86,HenTei86}). These theories have attracted interest in recent years, motivated, in particular, by supergravity and string theory. We will therefore consider null fields with arbitrary $n\ge3$ and $p$ (with $1\le p\le n-1$ to avoid trivial cases).
The relevant generalization of the concept of null fields is straightforward: all the {zeroth-order} scalar polynomial invariant constructed out of a $p$-form $\bF$ vanish (thus generalizing \eqref{null}) iff \cite{Hervik11} $\bF$ is of type N in the null alignment classification of \cite{Milsonetal05} (cf. Corollary~\ref{lemma_VSI0} below; for the case $p=2$ this was proven earlier in \cite{Coleyetal04vsi}). For this reason, in this paper  we shall use the terminology ``null'' and ``type N'' interchangeably, when referring to $p$-forms.

It will be also useful to observe that the type N condition of \cite{Milsonetal05} for $p$-forms can be easily rephrased in a manifestly frame-independent way, which we thus adopt as a definition here (see also \cite{Sokolowskietal93} when $p=2$):
\begin{definition}[$p$-forms of type N]
	\label{def_N}
 At a spacetime point, a $p$-form $\bF$ is of type N if it satisfies
\be
	\ell^a F_{a b_1\ldots b_{p-1}}=0 , \qquad \ell_{[a}F_{b_1\ldots b_{p}]}=0 ,
	\label{BelDeb}
\ee
where $\bl$ is a null vector (this follows from \eqref{BelDeb} and need not be assumed). The second condition can be equivalently replaced by $\ell^a\, {}^*F_{a b_1\ldots b_{n-p-1}}=0$ (cf. \cite{syngespec,Hallbook} for $n=4$, $p=2$).
\end{definition}

The most general algebraic form of a null $p$-form $\bF$ is thus known (eq.~\eqref{F_N} below). We will determine what are the necessary and sufficient conditions for a $p$-form $\bF$ living in a certain spacetime to be VSI, in the sense of the following definition:

\begin{definition}[VSI tensors]
	\label{def_VSI}
 A tensor in a spacetime with metric $g_{ab}$ is VSI$_I$ if {the} scalar polynomial invariants constructed from the tensor {itself} and its covariant derivatives up to order $I$ {($I=0,1,2,3,\dots$)} vanish. It is VSI if all its scalar polynomial invariants of {\em arbitrary order} vanish.  As in {\cite{Coleyetal04vsi,Pelavasetal05}}, if the Riemann tensor of $g_{ab}$ is VSI (or VSI$_I$), the spacetime itself is said to be VSI (or VSI$_I$).
\end{definition}

For the purposes of the present paper, it will be convenient to recall Corollary~3.2 of \cite{Hervik11}, which can be expressed as
\begin{theorem}[Algebraic VSI theorem \cite{Hervik11}]
 \label{alg_theor}
	A tensor is VSI$_0$ iff it is of type III (or more special).
\end{theorem}

Again, the tensor type refers to the algebraic classification of \cite{Milsonetal05} (see also the review \cite{OrtPraPra13rev}). Now, recalling that a non-zero $p$-form $\bF$ can only be of type G, II {(D)} or N ({cf. \cite{Durkeeetal10,Coleyetal04vsi,HerOrtWyl13}}),\footnote{However, it is well-known that for $p=2$ and $n=4$ the only possible types are D and N \cite{Stephanibook,Hallbook}. Note also that the type G does not occur for $p=2$ and $n$ even \cite{Sokolowskietal93,BerSen01,Milson04}. {For $p=1$ the type is G, D or N when the vector $\bF$ is timelike, spacelike or null, respectively (cf., e.g., \cite{OrtPraPra13rev}; in particular, $p=1$ is the only case of interest when $n=3$)}.} it immediately follows from Theorem~\ref{alg_theor} that 
\begin{cor}[VSI$_0$ $p$-forms] 
 \label{lemma_VSI0}
	A $p$-form $\bF$ is VSI$_0$ iff it is is of type N. 
\end{cor}

\subsection{Main result}

The main result of this paper {(proven in appendix \ref{app_proof})} is the following
\begin{theorem}[VSI $p$-forms]
 \label{theor}
The following two conditions are equivalent:

\begin{enumerate}

\item\label{cond1} a non-zero $p$-form field $\bF$ is VSI in a spacetime with metric $g_{ab}$
\item\label{cond2} 
 \begin{enumerate}
		\item\label{cond2a} $\bF$ possesses a multiple null direction $\bl$, i.e., it is of type N
		\item\label{cond2b} {$\pounds_{\bl}\bF=0$} 
		\item\label{cond2c} $g_{ab}$ is a {\em degenerate Kundt} metric, and $\bl$ is the corresponding Kundt null direction.
 \end{enumerate}

\end{enumerate}

\end{theorem}

We can already make a few observations about the implications of Theorem~\ref{theor}.

\begin{remark}[Degenerate Kundt metrics] 
	The definition of degenerate Kundt metrics \cite{ColHerPel09a,Coleyetal09} is reproduced in Appendix~\ref{app_Kundt_II} (Definition~\ref{def_deg}). By Proposition~\ref{prop_Kundt_deg} (with Propositions~7.1 and 7.3 of \cite{OrtPraPra13rev}), it follows that, e.g., all VSI spacetimes, all \pp waves, and all Kundt Einstein (in particular, Minkowski and (A)dS) or {aligned} pure radiation spacetimes necessarily belong to this class. As a consequence, when considering VSI $p$-forms {\em coupled to gravity}, the ``degenerate'' part of condition~\ref{cond2c} above becomes automatically trivially true (since a null $p$-form gives rise to {an aligned} pure radiation term in the energy-momentum tensor, cf. section~\ref{sec_EM}). It is also worth remarking that, in four dimensions, the degenerate Kundt spacetimes are the only metrics not determined by their curvature scalar invariants \cite{ColHerPel09a}, and they are thus of particular relevance for the equivalence problem \cite{ColHerPel09a,Coleyetal09}. 
	\label{rem_deg}
\end{remark}

\begin{remark}[{VSI vector field}]
Condition~\ref{cond2c} of Theorem~\ref{theor} implies that any {\em affinely parametrized} principal null vector $\bl$ of a VSI  $\bF$ is itself VSI.  Indeed, in the special case $p=1$, thanks to condition~\ref{cond2c}, condition~\ref{cond2a} is trivially satisfied and condition~\ref{cond2b} simply means that $\bl$ is affinely parametrized; i.e., for $p=1$, Theorem~\ref{theor} reduces to: {\em a vector field $\bl$ is VSI in a spacetime with metric $g_{ab}$ iff $\bl$ is {Kundt and affinely parameterized}, and $g_{ab}$ is a degenerate Kundt metric w.r.t $\bl$.} 
\label{rem_p=1}
\end{remark}

\begin{remark}[{Theory independence of the result}]
	In Theorem~\ref{theor}, the $p$-form $\bF$ is not assumed to satisfy any particular field equations and the result is thus rather general. On the other hand, {\em if} $\bF$ is taken to be  {\em closed} (i.e., $\d\bF=0$) 
	then condition~\ref{cond2b} automatically follows {from the type N condition \ref{cond2a}}, and need not be assumed. Otherwise, condition~\ref{cond2b} is needed to ensure (together with \ref{cond2a} and \ref{cond2c}) that $\nabla\bF$ is of type III (or more special). {Note that $\pounds_{\bl}\bF=\nabla_{\bl}\bF$ if $\bF$ is of type N and $\bl$ is Kundt.}
  \label{rem_theor_ind}
\end{remark}

\begin{remark}[{VSI$_3\Rightarrow$VSI for a $p$-form}]
	From the proof of Theorem~\ref{theor} (Appendix~\ref{app_proof}) it follows, in fact, that if $\bF$ is VSI$_3$ then it is necessarily VSI. Condition~\ref{cond1} could thus be accordingly relaxed in the proof ``\ref{cond1}. $\Rightarrow$ \ref{cond2}.''. By contrast, recall that in the case of the Riemann tensor one has  $VSI_2\Rightarrow VSI$ {\cite{Pravdaetal02,Coleyetal04vsi,Pelavasetal05}}. For completeness, necessary and sufficient conditions for $\bF$ to be VSI$_1$ or VSI$_2$ are given in Appendix~\ref{app_VSI_1_2}. 
	\label{rem_VSI}
\end{remark}

\begin{remark}[$\epsilon$-property]
 From Lemma~\ref{lemma_VSI} (Appendix~\ref{app_proof}), it is easy to see that a tensor is VSI if and only if, given an arbitrary non-negative integer $N$, there exist a reference frame in which its components and those of its covariant derivatives up to order $N$ can be made as small as desired (the proof of this statement is essentially the same as the one given in \cite{Pelavasetal05} for the Riemann tensor). This applies, in particular, to a $p$-form $\bF$ and is an extension of the early observations of \cite{Schroedinger35,syngespec} for the case $N=0$.
	\label{rem_epsilon}
\end{remark}

In the rest of the paper we discuss {further} implications of Theorem~\ref{theor} from a physical viewpoint. In section~\ref{sec_VSI} we give the explicit form {of VSI $p$-forms}, the associated (degenerate Kundt) background metric and the corresponding Maxwell equations in adapted coordinates. We also observe that null (and thus VSI) Maxwell fields are ``immune'' to adding a Chern-Simons term to the Maxwell equations (except when it is linear, i.e., for $n=2p-1$). More generally, we point out that certain VSI Maxwell fields solve {\em any} generalized electrodynamics. In section~\ref{sec_EM} we consider VSI $p$-forms in the full Einstein-Maxwell theory (which solve also certain supergravities due to the vanishing of the Chern-Simons term). We discuss consequences of the Einstein equations, also mentioning a few examples, and we comment on various subclasses of exact solutions (such as VSI spacetimes and \pp waves) with arbitrary $n$ and $p$. We observe that certain VSI Maxwell fields solve any generalized electrodynamics also when the coupling to gravity is kept into account. 
Appendix~\ref{app_D_Kundt} contains a summary of properties of Kundt spacetimes which are useful in this paper. {Appendix~\ref{app_proof} gives the proof of Theorem~\ref{theor}. For completeness, in Appendix~\ref{app_VSI_1_2} we present the necessary and sufficient conditions for a $p$-form to be VSI$_1$ or VSI$_2$, which can be seen as an ``intermediate'' result between {Corollary}~\ref{lemma_VSI0} and Theorem~\ref{theor}.}

\subsubsection*{Notation}

In an $n$-dimensional spacetime we set up a frame of $n$ real vectors $\bm{a} $ which consists of two null vectors $\bl\equiv{\mbox{\boldmath{$m_{(0)}$}}}$,  $\bn\equiv{\mbox{\boldmath{$m_{(1)}$}}}$ and $n-2$ orthonormal spacelike vectors $\bm{i} $, with $a, b\ldots=0,\ldots,n-1$ while $i, j  \ldots=2,\ldots,n-1$ \cite{Pravdaetal04} (see also the review \cite{OrtPraPra13rev} and references therein). For indices $i, j, \ldots$ there is no need to distinguish between subscripts and superscripts. 
The Ricci rotation coefficients $L_{ab}$, $N_{ab}$ and $\M{i}{a}{b}$ are defined by \cite{Pravdaetal04}
\be
 L_{ab}=\ell_{a;b} , \qquad N_{ab}=n_{a;b}  , \qquad \M{i}{a}{b}=m^{(i)}_{a;b} ,
 \label{Ricci_rot}
\ee
and satisfy the identities \cite{Pravdaetal04}
\be 
 L_{0a}=N_{1a}=N_{0a}+L_{1a}=\M{i}{0}{a} + L_{ia} = \M{i}{1}{a}+N_{ia}=\M{i}{j}{a}+\M{j}{i}{a}=0 . 
 \label{const-scalar-prod}
\ee

Covariant derivatives along the frame vectors are denoted as
\be
D \equiv \ell^a \nabla_a, \qquad \T\equiv n^a \nabla_a, \qquad \delta_i \equiv m^{(i)a} \nabla_a . 
 \label{covder}
\ee

\section{Explicit form of VSI electromagnetic fields and adapted coordinates}

\label{sec_VSI}

In this section we study VSI Maxwell {\em test} fields, i.e., without taking into account their backreaction on the spacetime geometry (which will be discussed in section~\ref{sec_EM}).

\subsection{Electromagnetic field and spacetime metric}

\label{subsec_explicit}

It is useful to express explicitly the conditions~\ref{cond2} of Theorem~\ref{theor} in a null frame adapted to $\bl$ (but otherwise arbitrary), as defined above. Condition~\ref{cond2a} reads {(cf., e.g., \cite{Durkeeetal10})}
\be
 F_{a b_1\ldots b_{p-1}}=p!F_{1i_1\ldots i_{p-1}}\ell_{[a}m^{(i_1)}_{\,b_1}\ldots m^{(i_{p-1})}_{\,b_{p-1}]} .
 \label{F_N}
\ee

The ``Kundt part'' of condition \ref{cond2c} means \eqref{Kundt}, i.e.,
\be
 L_{i0}=0, \qquad L_{ij}=0 . 
 \label{kundt1}
\ee
Conditions \ref{cond2b} and the remaining part of condition \ref{cond2c} are more conveniently represented in a null frame {\em parallelly transported along $\bl$} (i.e., such that \eqref{Kundt_rot} holds), where they take the form (cf. also appendix~\ref{app_D_Kundt}, and recall that a Kundt metric for which the Riemann tensor and its {\em first} covariant derivative are of aligned type II is necessarily degenerate Kundt \cite{ColHerPel09a,Coleyetal09})
\be
 DF_{1i_1\ldots i_{p-1}}=0 , \qquad R_{010i}=0 , \qquad DR_{0101}=0 , \qquad \mbox{with \eqref{Kundt_rot}}.
	\label{DF_DR}
\ee

In adapted coordinates, degenerate Kundt metrics are described by \cite{ColHerPel09a,Coleyetal09}

\be
 \d s^2 =2\d u\left[\d r+H(u,r,x)\d u+W_\alpha(u,r,x)\d x^\alpha\right]+ g_{\alpha\beta}(u,x) \d x^\alpha\d x^\beta , \label{Kundt_gen}
\ee
where $\bl=\partial_r$ is the Kundt vector, $\alpha,\beta=2 \dots n-1$, $x$ denotes collectively the set of coordinates $x^\alpha$, and $W_{\alpha,rr}=0=H_{,rrr}$ (thanks to which the second and third of~\eqref{DF_DR} are identically satisfied), i.e.,
\beqn
 & & W_{\alpha}(u,r,x)=rW_{\alpha}^{(1)}(u,x)+W_{\alpha}^{(0)}(u,x) , \label{deg_Kundt1} \\
 & & H(u,r,x)=r^2H^{(2)}(u,x)+rH^{(1)}(u,x)+H^{(0)}(u,x) . \label{deg_Kundt2}
\eeqn
In these coordinates one has (cf. the coordinate-independent expression \eqref{dl})
\be
  \ell_{a;b}{\d x^a\d x^b}=(2rH^{(2)}+H^{(1)})\d u^2+\frac{1}{2}W_{\alpha}^{(1)}(\d u\,\d x^\alpha+\d x^\alpha\d u).
 \label{l_deriv}
\ee

 The corresponding VSI $p$-form \eqref{F_N} reads
\be
 \bF=\frac{1}{(p-1)!}f_{\alpha_1\ldots\alpha_{p-1}}(u,x)\d u\wedge\d x^{\alpha_1}\wedge\ldots\wedge\d x^{\alpha_{p-1}} ,
 \label{F_N_coords}
\ee
where $f_{\alpha_1\ldots\alpha_{p-1}}\equiv F_{u\alpha_1\ldots\alpha_{p-1}}$ is $r$-independent due to the first of~\eqref{DF_DR}. In these coordinates we have (with the definition~\eqref{F2})
\be
	\F^2=f_{\alpha_1\ldots\alpha_{p-1}}f^{\alpha_1\ldots\alpha_{p-1}} ,
\ee	
where, from now on, it is understood that the indices of $f^{\alpha_1\ldots\alpha_{p-1}}$ are raised using the transverse metric $g^{\alpha\beta}$. $\F^2$ parametrizes the field strength of $\bF$ and is invariant under Lorentz transformations preserving $\bl$, as well as under transformations of the spatial coordinates $x\mapsto x'(x)$.

At this stage, $g_{\alpha\beta}$, $W_{\alpha}^{(1)}$, $W_{\alpha}^{(0)}$, $H^{(2)}$, $H^{(1)}$, $H^{(0)}$ and $F_{u\alpha_1\ldots\alpha_{p-1}}$ are all arbitrary functions of $u$ and $x$ (recall that {the invariant condition $W_{\alpha}^{(1)}=0\Leftrightarrow L_{i1}=0$ defines a special subfamily of Kundt metrics, cf. Appendix~\ref{app_Kundt_gen}}). In general, the associated Weyl and Ricci tensor are both of aligned type II (as follows from {Definition~\ref{def_deg} in Appendix~\ref{app_Kundt_II}}). {The} b.w.~0 components of all curvature tensors (and thus also their curvature invariants, of all orders) of the metric~\eqref{Kundt_gen} are independent of the functions  $W_{\alpha}^{(0)}$, $H^{(1)}$ and $H^{(0)}$ \cite{ColHerPel10} (as summarized in Proposition~7.2 of \cite{OrtPraPra13rev}). Restrictions coming from the Einstein equations are described below in section~\ref{sec_EM}.
We emphasize here that although the $p$-form \eqref{F_N_coords} is VSI, the spacetime \eqref{Kundt_gen} (with \eqref{deg_Kundt1}, \eqref{deg_Kundt2}) in general is not (not even VSI$_0$) \cite{Pravdaetal02,Coleyetal04vsi,Coleyetal06}, i.e., it may admit some non-zero invariants constructed {from} the Riemann tensor and its derivatives. However, all mixed invariants (i.e., those involving Riemann {\em and} $\bF$ together, {along with} their derivatives) are necessarily zero (since $\bF$ possesses only negative  b.w.s and Riemann only non-positive ones, and similarly for their derivatives.)

The metric~\eqref{Kundt_gen} includes, in particular, spacetimes of constant curvature (Minkowski and (A)dS). Therefore, the {$p$-form \eqref{F_N_coords}} can be used to describe {VSI} test fields in such backgrounds, as a special case.

\subsection{Maxwell's equations}

The construction of VSI $p$-forms has been, so far, purely geometric. Using \eqref{Kundt_gen} and \eqref{F_N_coords}, the source-free Maxwell equations $\d\bF=0$ and {$\d^*\bF=0$ reduce, respectively, to
\be
   f_{[\a_2\ldots\a_{p-1},\a_1]}=0 , \qquad  (\sqrt{\tilde g}\,f^{\beta\a_1\ldots\a_{p-2}})_{,\beta}=0 , 
	\label{Maxwell}
\ee
where $\tilde g\equiv\det g_{\alpha\beta}=-\det g_{ab}\equiv -g$. Effectively, these are Maxwell's equations for the $(p-1)$-form $\bff$ in the $(n-2)$-dimensional Riemannian geometry associated with $g_{\alpha\beta}$. 
In other words, the most general VSI $\bF$ that solves Maxwell's equations is given by \eqref{F_N_coords} in the spacetime \eqref{Kundt_gen} with \eqref{deg_Kundt1}, \eqref{deg_Kundt2}, where $\bff$ is {\em harmonic} w.r.t. $g_{\alpha\beta}$.\footnote{Note that eqs.~\eqref{Maxwell} apply to type N Maxwell fields $\bF$ in {\em any} Kundt spacetime, so also to type N $\bF$ that are not VSI (i.e., we have not used \eqref{deg_Kundt1} and \eqref{deg_Kundt2} to obtain \eqref{Maxwell}). In the case $p=2$, this agrees with the results of \cite{PodZof09}, once specialized to Maxwell fields of type N.} Recall, however, that $\bff$ can also depend on $u$.

\subsubsection{{The special case $p=1$}}

\label{subsubsec_p=1}

{The case $p=1$ (or $p=n-1$ by duality), for which $\bF=\bl$ is a vector field, is special. Indeed by Theorem~\ref{theor} (see also Remark~\ref{rem_p=1}), if $\bl$ is VSI then it must be degenerate Kundt and affinely parametrized. Further, by a boost, $\bl$ can always be rescaled  (still remaining VSI) such that $L_{1i}=L_{i1}$ (cf., e.g., \cite{Coleyetal04vsi,HerPraPra14}). It is then easy to see that Maxwell's equations are satisfied (i.e., $\ell^a_{\ ;a}=0=\ell_{[a;b]}$, cf.~\eqref{dl}). Therefore, {\em for $p=1$ to any VSI $\bF$ it can always be associated a solution to the sourcefree Maxwell equations} (in this statement, ``VSI'' can also be relaxed to ``VSI$_1$'', since only the Kundt property of $\bl$ has been employed, cf. also Proposition~\ref{proposition_VSI1_2})
}

\subsubsection{The special cases $n=4,3$}

\paragraph{Case $n=4$}

We have seen (Theorem~\ref{theor}) that, for any $n$ and $p$, if a $p$-form $\bF$ is VSI then the multiply aligned null direction $\bl$ is necessarily (degenerate) Kundt and thus, in particular, geodesic and shearfree. In the case $n=4$, $p=2$, the Robinson theorem \cite{robinsonnull,Stephanibook,penrosebook2} then implies that the family of null bivectors associated with $\bl$ includes a solution to the sourcefree Maxwell equations.  Namely, from a VSI $\bF$ one can always obtain a solution $\bF'$ to the Maxwell equations by means of a Lorentz transformation (spin and boost) preserving the null direction defined by $\bl$ (of course, it may happen that $\bF$ is already a solution and no transformation is needed). The $2$-form $\bF'$ will thus be also multiply aligned with $\bl$ and therefore it will be still VSI (condition 2b of Theorem~\ref{theor} will hold due to $\d \bF'=0$,  cf. Remark \ref{rem_theor_ind}). {Along with the observation in section~\ref{subsubsec_p=1} for the case $p=1$ (or $p=3$), one concludes that} for $n=4$, given a VSI field, one can always associate with it a VSI solution to the sourcefree Maxwell equations ({as above, it suffices to assume that $\bF$ is VSI$_1$}). This does not seem to be true when $n>4$, in general.

\paragraph{Case $n=3$}

In three dimensions the situation is even simpler, {since the only case of interest is $p=1$ (dual to $p=2$). The observation in section~\ref{subsubsec_p=1} thus immediately implies that} also {\em for $n=3$ to any VSI {(or VSI$_1$)} $\bF$ it can always be associated a solution to the sourcefree Maxwell equations}.

\subsection{Maxwell-Chern-Simons' equations}

\label{subsec_CS}

Maxwell's equations \eqref{Maxwell} for a $p$-form field admit a generalization which includes a Chern-Simons term, i.e., $\d\bF=0$ and $\d^*\bF+\alpha\bF\wedge\ldots\wedge\bF=0$, where $\alpha\neq0$ is an arbitrary constant. The second term in the latter equation contains $k$ factors $\bF$, and the corresponding number of spacetime dimensions is given by $n=p(k+1)-1$. Such modifications of Maxwell's equations appear, for example, in (the bosonic sector of) minimal supergravity in five and eleven dimensions (with  $n=5$, $p=2$, $k=2$ and $n=11$, $p=4$, $k=2$, respectively -- cf., e.g., \cite{Ortinbook} and references therein).

\subsubsection{Generic case $k\ge2$}

\label{subsubsec_CS_generic}

Now, let us note that any $p$-form of type N \eqref{F_N} satisfies
\be
	\bF\wedge\bF=0 ,
\ee
so that {\em for a $\bF$ of type N, the Chern-Simons term vanishes identically}, provided $k\ge2$. Therefore, a type N solution $\bF$ of Maxwell's theory \eqref{Maxwell} is automatically also a solution of Maxwell-Chern-Simons' theory. This applies, in particular, to VSI solutions. This has been also noticed, e.g., in \cite{FigPap01} in the case of certain VSI \pp waves in $n=11$ supergravity.

\subsubsection{Special case $k=1$}

\label{subssubsec_CS_k=1}

The special case $k=1$ results in a linear theory 
\be
 \d\bF=0 , \qquad \d^*\bF+\alpha\bF=0 \qquad (n=2p-1) .
 \label{CS_k=1} 
\ee
It is clear that here the Maxwell equations are modified non-trivially also for type N fields. The second of \eqref{Maxwell} now has to be replaced by 
\be
  (\sqrt{\tilde g}\,f^{\beta\a_1\ldots\a_{p-2}})_{,\beta}-\alpha\sqrt{\tilde g}\star\! f^{\a_1\ldots\a_{p-2}}=0 ,
\ee
where $\star$ is the Hodge dual in the transverse geometry $g_{\alpha\beta}$ (not to be confused with the $n$-dimensional Hodge dual $^*$ in the full spacetime $g_{ab}$). {Recall that a linear Chern-Simon term appears, e.g., in  topological massive electrodynamics in three dimensions ($n=3$, $p=2$, $k=1$) {\cite{Schonfeld81,DesJacTem82} (and references therein)}.

\subsection{Universal solutions of generalized electrodynamics (test fields)}

\label{subsec_univ_test}

Theories of electrodynamics described by a Lagrangian depending also on the derivatives of the field strength have been proposed long ago in \cite{Bopp40,Podolsky42}. More recently, interest in higher-derivative theories has been also motivated by string theory, cf., e.g., \cite{AndTse88,Thorlacius98,CheKalOrt12} and references therein. In this context, the electromagnetic field is typically represented by a closed 2-form $\bF$ whose field equations contain ``correction terms'' constructed in terms of $\bF$ and its covariant derivatives.

The class of VSI Maxwell fields defined in this paper can be employed to identify a subset of VSI solutions that are ``universal'', i.e., solving simultaneously any electrodynamics  whose field equations can be expressed as $\d\bF=0$, $*\d\!*\!\!\tilde\bF=0$, where $\tilde\bF$ can be any $p$-form constructed from $\bF$ and its covariant derivatives.\footnote{We assume that $\tilde\bF$ is constructed polynomially from these quantities (however, the scalar coefficients appearing in such polynomials need not be polynomials of the scalar invariants of $\bF$ and its covariant derivatives -- cf., for instance, Born-Infeld's theory). The same type of construction will be assumed for the energy-momentum tensor of generalized theories in section~\ref{subsec_univ_EM} (but see also footnote~\ref{foot_UEM}).} {It can be shown, for example, that  any VSI Maxwell $\bF$ is universal if the background is a Kundt spacetime of Weyl and traceless-Ricci type III (aligned) with $DR=0=\delta_i R$. In particular, Ricci flat and Einstein Kundt spacetimes of Weyl type III/N/O {can occur} (an explicit example is given in section~\ref{subsec_univ_EM}), the latter including Minkowski and (A)dS.}  Details and more general examples (also of Weyl type II) will be presented elsewhere.

It should be emphasized that the above definition of universality includes terms with arbitrary higher-order derivative corrections. If one restricts oneself to theories in which $\tilde\bF$ is constructed algebraically from $\bF$, i.e., without taking derivatives of  $\bF$ (like NLE, for which {$n=2p=4$} and the modified Maxwell equations are of the form $\d(f_1{}^*\bF+f_2\bF)=0$), then these admit as solutions {\em all} null solutions of Maxwell's equations (not necessarily VSI), without any restriction on the background geometry. This has been known for a long time {in NLE} \cite{Schroedinger35,Schroedinger43} (see also, e.g., \cite{BicSla75}).

\section{Einstein-Maxwell solutions}

\label{sec_EM}

\subsection{General equations}
\label{sec_GE}

What discussed so far applies to VSI test fields, since we have not considered the consequences of the backreaction on the spacetime geometry. In the full Einstein-Maxwell theory this is described by the energy-momentum tensor associated with $\bF$ 
\begin{equation}\label{Energy momentum}
  T_{ab} = \frac{\b}{8\pi} \left(F_{ac_1\ldots c_{p-1}} {F_b}^{c_1\ldots c_{p-1}}-\frac{1}{2p} g_{ab}F^2 \right) ,
\end{equation}
where $F^2=F_{\a_1\ldots\a_{p}}F^{\a_1\ldots\a_{p}}$. With \eqref{F_N}, $T_{ab}$ in \eqref{Energy momentum} {takes the form of} aligned pure radiation, and Einstein's equations with a cosmological constant $R_{ab} - \frac{1}{2} R g_{ab} + \Lambda g_{ab} = 8 \pi T_{ab}$ {reduce to}
\begin{equation}\label{Ricci}
	R_{ab} = \frac{{2\Lambda}}{n-2} g_{ab} + \b\F^2\ell_a\ell_b .
\end{equation}

The Ricci tensor {of \eqref{Kundt_gen}} thus must satisy $R_{01}=\frac{2}{n-2} \Lambda$, $R_{ij}=\frac{2}{n-2} \Lambda\delta_{ij}$, $R_{1i}=0$ and $R_{11}=\b\F^2$ (with $R=2n\Lambda/(n-2)$, while $R_{00}=0=R_{0i}$ identically since the Riemann type is II by construction). Using these, the b.w.~0 components of the Einstein equations imply that in {\eqref{Kundt_gen}--\eqref{deg_Kundt2}}  (as follows readily from \cite{PodZof09})
\beqn
 & & \R_{\alpha\beta}=\frac{2\Lambda}{n-2}g_{\alpha\beta}+\frac{1}{2}W_{\alpha}^{(1)}W_{\beta}^{(1)}-W_{(\alpha||\beta)}^{(1)} , \label{Rij} \\
 & & 2H^{(2)}=\frac{\R}{2}-\frac{n-4}{n-2}\Lambda+\frac{1}{4}W^{(1)\alpha}W_{\alpha}^{(1)} , \label{H2} 
\eeqn
where $\R_{\alpha\beta}$, $\R$ and $||$ denote, respectively, the Ricci tensor, the Ricci scalar and the covariant derivative associated with $g_{\alpha\beta}$, and $W^{(1)\alpha}\equiv g^{\alpha\beta}W_{\beta}^{(1)}$. The first of these is an ``effective Einstein equation'' for the transverse metric, while the second one determines the function $H^{(2)}$. Note that the functions $W_{\alpha}^{(0)}$, $H^{(1)}$ and $H^{(0)}$ do not appear here. 

With \eqref{Rij}, the contracted Bianchi identity in the transverse geometry, i.e., ${2\R_{\alpha\beta}}^{||\B}=\R_{,\a}$, tells us  (after simple manipulations with the Ricci identity) that $W_{\a}^{(1)}$ is constrained by
\be
 W_{\a||\B}^{(1)\ \B}=\frac{1}{2}W^{(1)\B}\left(3W_{\a||\B}^{(1)}-W_{\B||\a}^{(1)}\right)+W^{(1)}_{\a}\left(W^{(1)\B}_{\ \ \ \ \ ||\B}-\frac{1}{2}W^{(1)\B}W^{(1)}_\B -\frac{2 \Lambda}{n-2}\right) .
\ee

Finally, the Einstein equations of negative b.w., which can be used to determine the functions $H^{(1)}$ and $H^{(0)}$, reduce to \cite{PodZof09}\footnote{Eqs.~\eqref{H1} and \eqref{H0} are equations (76) and (73) of \cite{PodZof09} (once the non-null part of the electromagnetic field is set to zero there, and up to minor reshuffling and different notation). Note that, {at least in the case of an aligned null $\bF$,} the remaining Einstein equations (71), (72) and (75) given in \cite{PodZof09} are identically satisfied as a consequence of \eqref{Rij}--\eqref{H0} of the present paper, and can thus be dropped ({it is plausible that this remains true also for an aligned non-null $\bF$, as known when $n=4$ \cite{Stephanibook}, and for $n>4$ in the special case $W_{\alpha}^{(1)}=0$ \cite{Krtousetal12} -- we have not tried to verify it since it is irrelevant to the present paper}). To verify this statement one needs repeated use of the Ricci identity and standard identities such as $g_{ab,c}=g_{ad}\Gamma^{d}_{\ bc}+g_{bd}\Gamma^{d}_{\ ac}$ and $\Gamma^{b}_{\ ab}=(\ln\sqrt{-g})_{,a}=\frac{1}{2}g^{bc}g_{bc,a}$. In the case of {(72,\cite{PodZof09})} one also needs the identity $g_{\a\B,u}^{\ \ \ \ \ ||\a\B}-g^{\a\B}g_{\a\B,u||\gamma}^{\ \ \ \ \ \ \ \ \gamma}=g^{\a\B}(\tilde\Gamma^\gamma_{\ \a\B,u||\gamma}-\tilde\Gamma^{\gamma}_{\ \gamma\B,u||\a})=g^{\a\B}\R_{\a\B,u}$ (where $\tilde\Gamma^\gamma_{\ \a\B,u||\gamma}$ are the Christoffel symbols of $g_{\a\B}$ -- in the last equality the (contracted) standard definition of the Riemann tensor has been employed).\label{foot_PZ}}
\beqn
 & & 2H^{(1)}_{,\a}=-{g_{\a\B,u}}^{||\B}+2W^{(0)\ \ \ \B}_{[\a||\B]}-2W^{(0)\B}W^{(1)}_{\a||\B}+(W^{(0)\B}W^{(1)}_\B)_{,\a}+W^{(1)}_{\a,u}+2(\ln\sqrt{\tilde g})_{,u\a} \nonumber \\
 & & \qquad\qquad {}+W^{(1)}_\a\left[W^{(0)\B}W^{(1)}_\B-W^{(0)\B}_{\ \ \ \ \ ||\B}+(\ln\sqrt{\tilde g})_{,u}\right]+\frac{4\Lambda}{n-2}W^{(0)}_\a , \label{H1} \\
 & & \Delta H^{(0)}+W^{(1)\a}H^{(0)}_{,\a}+W^{(1)\a}_{\ \ \ \ \ ||\a}H^{(0)}=W^{(0)\B}W^{(0)}_\B\left(\frac{1}{2}W^{(1)\a}_{\ \ \ \ \ ||\a}-\frac{2\Lambda}{n-2}\right) \nonumber \\
 & & \qquad\qquad {}+H^{(1)}\left[W^{(0)\a}_{\ \ \ \ \ ||\a}-(\ln\sqrt{\tilde g})_{,u}\right]-\frac{1}{2}(W^{(0)\a}W^{(1)}_\a)^2+W^{(0)[\a||\B]}W^{(0)}_{[\a||\B]}+W^{(0)\ ||\a}_{\a,u} \nonumber \\
 & & \qquad\qquad {}-W^{(0)\B}\left(2W^{(1)\a}W^{(0)}_{[\a||\B]}+W^{(1)}_{\B,u}-2H^{(1)}_{,\B}\right)-(\ln\sqrt{\tilde g})_{,uu}+\frac{1}{4}g^{\a\B}_{\ \ ,u}g_{\a\B,u} -\b\F^2 , \label{H0}
\eeqn
where $\Delta$ is the Laplace operator in the geometry of the transverse metric $g_{\a\B}$ (not to be confused with the symbol $\T$ defined in \eqref{covder}). In addition to \eqref{Rij}--\eqref{H0}, the equations for the electromagnetic field \eqref{Maxwell} must also be satisfied  (note that the only metric functions entering there are the $g_{\alpha\beta}$). {When $\F^2=0$, eqs.~\eqref{Rij}--\eqref{H0} represent the vacuum Einstein equations for the most general Kundt spacetime.}

Let us observe that the functions $\F^2(u,x)$ and $H^{(0)}(u,x)$ enter only {\eqref{H0}}, which is linear in  $H^{(0)}(u,x)$. It follows, e.g., that given any Kundt Einstein spacetime (necessarily of the form \eqref{Kundt_gen} with \eqref{deg_Kundt1}, \eqref{deg_Kundt2}), one can add to it an electromagnetic (and gravitational) wave by just appropriately choosing a new function $H^{(0)}(u,x)$, and leaving the other metric functions unchanged (amounting, in fact, to a generalized Kerr-Schild transformation; cf. Theorem~31.1 and section 31.6 of \cite{Stephanibook} in four dimensions). The resulting solution will describe a wave-like VSI $p$-form field propagating in the chosen Kundt Einstein spacetime. If the Einstein ``seed'' is VSI (or CSI), so will be the corresponding Einstein-Maxwell spacetime (see \cite{ColHerPel10} and the comments in section~\ref{subsec_explicit}).

The simplest such examples one can construct are electromagnetic and gravitational ``plane-fronted'' waves (with $W_{\alpha}^{(0)}=0$) propagating in a constant curvature background, giving rise to Kundt waves of Weyl type N. These are well-known in four dimensions for any value of $\Lambda$ \cite{GarPle81,OzsRobRoz85} (and include, e.g., the Siklos waves when $\Lambda<0$; see also \cite{Stephanibook,BicPod99I,GriPodbook}), and have been considered also in arbitrary dimensions \cite{Obukhov04} (for $p=2$, but a generalization to any $p$ is straightforward). Similarly, one can construct, e.g., electrovac waves of Weyl type II in (anti-)Nariai product spaces for $n=4$ \cite{PodOrt03} and higher \cite{Krtousetal12}. All these examples are CSI spacetimes, and for $\Lambda=0$ they become VSI ({\em pp}- or Kundt waves) spacetimes (cf. sections \ref{subsubsec_VSI} and \ref{subsubsec_pp}).
More general (e.g., with $W_{\alpha}^{(0)}\neq0$) degenerate Kundt metrics with null Maxwell fields are also known (see \cite{Stephanibook,GriPodbook,GriDocPod04} and references therein for $n=4$ and, e.g., \cite{CalKleZor07,Krtousetal12} in higher dimensions  -- more references in section~\ref{subsubsec_VSI} below in the case of VSI spacetimes).

In general, for all Einstein-Maxwell solutions with a VSI $p$-form $\bF$, since all the mixed invariants are zero (cf. section~\ref{subsec_explicit}) and since the Ricci tensor is constructed out of $\bF$ (eq.~\eqref{Ricci}), the only possible non-zero scalar invariants are those constructed from the Weyl tensor (and its derivatives), and the Ricci scalar $R$. We also remark that, thanks to the observations of section~\ref{subsec_CS}, all such Einstein-Maxwell solutions having $n=5$, $p=2$ or $n=11$, $p=4$ are also solutions (with the same $\bF$) of the bosonic sector of 5D minimal supergravity (ungauged or gauged) and 11D supergravity, respectively (further comments and references in section~\ref{subsubsec_VSI}).

\subsubsection{VSI spacetimes with VSI Maxwell fields}

\label{subsubsec_VSI}

The special case when the spacetime metric is VSI is of particular interest. All VSI metrics in HD (in particular, with Ricci type N) are given in \cite{Coleyetal06} (see also \cite{Coleyetal07}). In a VSI spacetime, one can always choose coordinates such that the metric is given by \eqref{Kundt_gen} with \cite{Coleyetal06}
\be
	g_{\alpha\beta}=\delta_{\alpha\beta} , \qquad W_{\alpha}^{(1)}=-\delta_{\alpha, 2} \frac{2 \epsilon}{x_2} , \qquad H^{(2)}=\frac{\epsilon}{2 (x^2)^2} \qquad (\epsilon=0,1) .
\ee

The VSI assumption implies $\Lambda=0$, so that the Einstein equations \eqref{Ricci} give $R_{ab} =\b\F^2\ell_a\ell_b$, i.e., the Ricci tensor is of aligned type N (and \eqref{Rij} and \eqref{H2} are satisfied identically). This constrains the functions $W_{\alpha}^{(0)}$, $H^{(1)}$ and $H^{(0)}$, as detailed in \cite{Coleyetal06}. The Weyl tensor is of type III aligned with $\bl=\pa_r$ (it becomes of type N for special choices of $W_{\alpha}^{(0)}$, including $W_{\alpha}^{(0)}=0$,\footnote{We note that, when the Weyl type is N, the metric takes the Kerr-Schild form -- the argument given in section~4.2.2 of \cite{OrtPraPra09} in the vacuum case holds also for Ricci type N. When the Weyl type is III, they are of the more general ``extended Kerr-Schild'' form \cite{Malek14}.} and of type O under further conditions on $H^{(0)}$ \cite{Coleyetal06}). 
The Maxwell field is still given by \eqref{F_N_coords}. Some VSI spacetimes (more general than VSI \pp waves) coupled to null $p$-forms have been discussed in \cite{Coleyetal07} in the context of type IIB supergravity (but note that a few of those are of Ricci type III due to the presence of an additional non-trivial dilaton\footnote{We did not consider a dilaton $\varphi$ in our discussion, but this can be easily included. The dilaton $\varphi$ itself cannot be VSI unless zero, being a scalar field. However, if we want $\varphi_{,a}$ to be VSI, then it must obviously be a null vector field. If we also require the mixed invariant ${F^a}_{c_1\ldots c_{p-1}} {F_b}^{c_1\ldots c_{p-1}}\varphi_{,a}\varphi_{,b}$ to vanish, we immediately obtain $\varphi_{,a}\propto \ell_a=(\d u)_a$, which implies $\varphi=\varphi(u)$ (as assumed in \cite{Guven87,HorSte90,Coley02,Coleyetal07}). This is also a sufficient condition for $\varphi_{,a}$ to be VSI (since $\varphi_{,a}=\varphi'\ell_a$ and $\bl$ is {degenerate Kundt, cf. Remark~\ref{rem_p=1}}) and for all mixed invariants (i.e., containing $\varphi$, the Riemann/Maxwell tensors {and their derivatives}) to vanish as well. However, $\varphi_{;ab}=\varphi''\ell_a\ell_b+\varphi'\ell_{a;b}$ is a symmetric rank-2 tensor (of aligned type III or more special) which contains also non-zero components of b.w. $-1$  iff $\varphi'W_{\alpha}^{(1)}\neq0$ (cf. \eqref{l_deriv}) {and appears on the r.h.s. of Einstein's equations \cite{Coleyetal07}, thus modifying the Ricci type}.\label{foot_dilat}}).

As a special subcase, for $\epsilon=0=H^{(1)}$ one obtains VSI \pp waves with a null $p$-form, for which the Weyl type can only be III(a) or more special (since the Ricci type is N \cite{Coleyetal06}). In 4D, these are well-known and necessarily of Weyl type N (which coincides with the type III(a) for $n=4$ \cite{OrtPraPra13rev}) or O  (cf. section~24.5 of \cite{Stephanibook} and references therein). In higher dimensions, their role in the context of supergravity and string theory has been known for some time, see, e.g., {\cite{KowalskiGlikman84,Hull84,Guven87,AmaKli89,HorSte90,Tseytlin93,BerKalOrt93}.} More recently, some of these (with $n$ arbitrary, $p=2,3$ and $W_{\alpha}=W_{\alpha}^{(0)}(u,x)\neq 0$) have been interpreted as ``charged gyratons'' \cite{FroZel06,FroLin06}.\footnote{It was indeed observed in \cite{FroZel06} that the ansatz used there leads to VSI spacetimes in which also all the electromagnetic scalar invariants vanish. This appears as a special subcase of the result given in Theorem~\ref{theor}.} From the viewpoint of supersymmetry, it is worth recalling also that a VSI spacetime admitting a timelike or null Killing vector field must necessarily be a VSI \pp wave (see the appendix of \cite{Coleyetal07}). Supersymmetric VSI \pp waves coupled to null $p$-forms are indeed well-known in various supergravities (see, e.g., \cite{Tod83,KowalskiGlikman84,Hull84,Guven87,BerKalOrt93,Gauntlettetal03} and references therein).

\subsubsection{\pp waves with VSI Maxwell fields}

\label{subsubsec_pp}

As mentioned above, spacetimes admitting a null Killing vector field are of special interest for supersymmetry. A null Killing vector field aligned with the Ricci tensor is necessarily Kundt (cf., e.g., Proposition~8.21 of \cite{OrtPraPra13rev}). Now, in the metric \eqref{Kundt_gen}, the Kundt vector $\bl=\pa_r$ is Killing iff  
\be
	W_{\alpha}^{(1)}=0 , \qquad H^{(2)}=0=H^{(1)} ,
\ee
which is equivalent to requiring \eqref{Kundt_gen} to be a \pp wave \cite{Brinkmann25}, i.e., $\ell_{a;b}=0$ {(cf. \eqref{l_deriv})}. 

For \pp waves the Einstein equations \eqref{Ricci} imply $\Lambda=0$ (cf., e.g., Proposition~7.3 of \cite{OrtPraPra13rev}), so that $R_{ab} =\b\F^2\ell_a\ell_b$ is of type N. By \eqref{Rij}, it follows that the transverse metric $g_{\alpha\beta}(u,x)$ must be Ricci-flat (if it is flat, as happens necessarily for $n=4,5$, one has VSI \pp waves, cf. section~\ref{subsubsec_VSI}), and \eqref{H2} is then identically satisfied.  Constrains on the functions $W_{\alpha}^{(0)}$ and $H^{(0)}$ follow from the remaining Einstein equations \eqref{H1}, \eqref{H0}, i.e.,
\beqn
 & & 2W^{(0)\ \ \ \B}_{[\a||\B]}={g_{\a\B,u}}^{||\B}-2(\ln\sqrt{\tilde g})_{,u\a} , \nonumber \\
 & & \T H^{(0)}=W^{(0)[\a||\B]}W^{(0)}_{[\a||\B]}+W^{(0)\ ||\a}_{\a,u}-(\ln\sqrt{\tilde g})_{,uu}+\frac{1}{4}g^{\a\B}_{\ \ ,u}g_{\a\B,u} -\b\F^2 . \label{H0_pp}
\eeqn

Similarly as for Ricci flat \pp waves (cf. Table~2 of \cite{Ortaggio09} and Proposition 7.3 of \cite{OrtPraPra13rev}), here the Weyl type is II'(abd) (which reduces to III(a) for $n=5$ and to N for $n=4$) or more special. The Maxwell field is given by \eqref{F_N_coords}. 

Some CSI (non-VSI) \pp waves coupled to null forms in Ricci-flat direct products arise as special cases of the solutions of \cite{DerGur86} (where $n=11$ and $p=4$).

\subsection{Universal Einstein-Maxwell solutions}

\label{subsec_univ_EM}

In section~\ref{subsec_univ_test} we commented on the role of a subset of the VSI Maxwell fields as {universal} solutions of {all} generalized electrodynamics on certain backgrounds (test fields). More generally, some of those can also be used to construct exact solutions of  full general relativity, i.e., keeping into account the backreaction of the electromagnetic field on the spacetime geometry. This is described by Einstein's equations in which, however, the $T_{ab}$ associated with the electromagnetic field is determined in the generalized electrodynamics (in terms of $\bF$ and its covariant derivatives, thus being generically different from \eqref{Energy momentum} {-- cf., e.g., \cite{Podolsky42}}). The class of universal Einstein-(generalized-)Maxwell solutions deserves a more detailed study, which we will present elsewhere. Here we only point out some examples. Namely, it can be shown that all VSI spacetimes with $L_{i1}=0=L_{1i}$ (i.e., the recurrent ones) coupled to an aligned VSI $p$-form field that solve the standard Einstein-Maxwell equations (and are thus of Ricci type N) are also exact solutions of gravity coupled to generalized electrodynamics,\footnote{To be precise, we should exclude from the discussion possible peculiar theories admitting an energy-momentum tensor $T_{ab}$ that vanishes for certain non-zero electromagnetic fields (or at least for the ``universal'' ones). One possible way to ensure this is, for example, to consider only theories for which $T_{ab}$ is of the form $T_{ab}\propto T_{ab}^{EM}+\mbox{``corrections''}$, where $T_{ab}^{EM}$ is the energy-momentum tensor of the standard Einstein-Maxwell theory, and the ``corrections'' are terms that go to zero faster that $T_{ab}^{EM}$ in the limit of weak fields. Additionally, it is also understood that a constant rescaling of a universal $p$-form $\bF$ (or, alternatively, of the corresponding metric, and so also the Ricci tensor) may be necessary when going from one theory to another.\label{foot_UEM}} provided $p>1$ and $\delta_iF_{1j_1\ldots j_{p-1}}=0$ (in an ``adapted'' parallely transported frame, i.e., such that $\M{i}{j}{k}=0$). Within this family, metrics of Weyl type N are necessarily \pp waves, for which such a universal property {was pointed out in} \cite{Guven87,HorSte90,Horowitz90}, at least for certain values of $p$ (\cite{Guven87} considered only plane waves, but included also Yang-Mills field). But metrics of Weyl type III are also permitted, including \pp waves ($L_{11}=0$) and also genuinely recurrent ($L_{11}\neq0$) spacetimes (for $n=4$, $p=3$ this was discussed in \cite{Coley02}). One {explicit} example of the latter solutions in 4D is given by\footnote{This solution has been obtained by adding a null $\bF$ to a type III vacuum spacetime found by Petrov (eq.~(31.40) in \cite{Stephanibook}), {cf. section~ \ref{sec_GE}}. A parallelly transported frame satisfying $\M{i}{j}{k}=0$ is given by $\bl=\pa_r$, $\bn=\pa_u-\frac{1}{2}\left(xr-xe^x-2\b e^xc^2(u)\right)\pa_r$, $\bbm_{2}=e^{-x/2}\left(\cos\frac{ye^u}{2}\pa_x-e^{-u}\sin\frac{ye^u}{2}\pa_y\right)$, $\bbm_{3}=e^{-x/2}\left(\sin\frac{ye^u}{2}\pa_x+e^{-u}\cos\frac{ye^u}{2}\pa_y\right)$. {Note that by setting $\b=0$ in \eqref{Petrov_EM}, this solution represents a universal {\em test} Maxwell field in a Petrov type III vacuum spacetime, relevant to the discussion in section~\ref{subsec_univ_test}.}} 
\beqn
 & & \d s^2 =2\d u\left[\d r+\frac{1}{2}\left(xr-xe^x-2\b e^xc^2(u)\right)\d u\right]+ e^x(\d x^2+e^{2u}\d y^2) , \label{Petrov_EM} \\
 & & \bF=e^{x/2}c(u)\d u\wedge\left(-\cos\frac{ye^u}{2}\d x+e^u\sin\frac{ye^u}{2}\d y\right) .
\eeqn

Similarly as in section~\ref{subsec_univ_test}, the above discussion applies to generalized electrodynamics with arbitrary higher-order derivative corrections. As a special case, the fact that Einstein-Maxwell solutions with aligned null electromagnetic fields (not necessarily VSI) are also solution of NLE coupled to gravity was previously demonstrated in \cite{Kichenassamy59,KreKic60,Peres61}.

\section*{Acknowledgments}

This work has been supported by research plan {RVO: 67985840} and research grant GA\v CR 13-10042S.

\renewcommand{\thesection}{\Alph{section}}
\setcounter{section}{0}

\renewcommand{\theequation}{{\thesection}\arabic{equation}}

\section{Some of the Newman-Penrose equations for Kundt spacetimes}
\setcounter{equation}{0}

\label{app_D_Kundt}

\subsection{General Kundt spacetimes}

\label{app_Kundt_gen}

By assumption $\bl$ is geodesic and Kundt (i.e., expansionfree, shearfree and twistfree), so that {(recall the definitions \eqref{Ricci_rot})} \cite{Pravdaetal04,Coleyetal04vsi,OrtPraPra07}
\be
 L_{i0}=0, \qquad L_{ij}=0  . 
 \label{Kundt}
\ee
Without loss of generality we can use {\em an affine parametrization and a frame parallelly transported along $\bl$}, such that, in addition to \eqref{Kundt}, we also have \cite{Pravdaetal04,Coleyetal04vsi,OrtPraPra07}
\be
 L_{10}=0, \qquad \M{i}{j}{0}=0, \qquad N_{i0}=0 . 
 \label{Kundt_rot}
\ee

Thanks to these, the covariant derivatives of the frame vectors take the form \cite{Pravdaetal04,Coleyetal04vsi}
\beqn
& & \ell_{a ; b }=L_{11} \ell_a \ell_b +L_{1i} \ell_a \msub{i}{b}+L_{i1} \msub{i}{a} \ell_b  \label{dl} , \\
& & n_{a ; b }=-L_{11} n_a \ell_b -L_{1i} n_a \msub{i}{b}+N_{i1} \msub{i}{a} \ell_b  + N_{ij} \msub{i}{a} \msub{j}{b} \label{dn} , \\
& & \msub{i}{a ; b }=-{N}_{i1} \ell_a \ell_b-{L}_{i1} n_a \ell_b-{N}_{ij} \ell_a \msub{j}{b}+\M{i}{j}{1} \msub{j}{a} \ell_b  + \M{i}{k}{l} \msub{k}{a} \msub{l}{b} . \label{dm} 
\eeqn

From the Ricci identities (11g) and (11k) of \cite{OrtPraPra07} with \eqref{Kundt}, \eqref{Kundt_rot} it follows immediately 
\be
 R_{0i0j}=0 , \qquad R_{0ijk}=0 ,
\label{RiemIb}
\ee
which implies that all Kundt spacetimes are of aligned Riemann type I, with the further restriction $R_{0ijk}=0$ (cf.~\cite{PodZof09,Coleyetal09}).
Furthermore, the Ricci identities (11b), (11e), (11n), (11a), (11j), (11m) and (11f) of \cite{OrtPraPra07} read
\beqn
 & & DL_{1i}=-R_{010i} , \qquad DL_{i1}=-R_{010i} , \label{11be} \\
 & & D\M{i}{j}{k}=0 , \label{11n} \\
 & & DL_{11}=-L_{1i} L_{i1}-R_{0101}, \label{11a} \\
 & & DN_{ij}=-R_{0j1i} \label{11j} , \\
 & &  D\M{i}{j}{1}=-\M{i}{j}{k}L_{k1}-R_{01ij} , \label{11m} \\
 & &  DN_{i1}=-N_{ij}L_{j1}+R_{101i} \label{11f} , 
\eeqn
while the commutators \cite{Coleyetal04vsi} needed in this paper simplify to
\beqn
& & \T D - D \T = L_{11} D + L_{i1} \delta_i ,  \label{TD} \\
& & \delta_i D - D \delta_i=L_{1i} D . \label{dD} 
\eeqn

In adapted coordinates, the general Kundt line-element can be written as \eqref{Kundt_gen}, where  $\bl=\partial_r$ and all the metric functions depend arbitrarily on their arguments. In those coordinates $L_{1i}=L_{i1}$ (since $\ell_a=(\d u)_a$), cf. \eqref{l_deriv}. From~\eqref{l_deriv}, \eqref{dl} it follows that $\bl$ is {\em recurrent} iff $W_{\alpha}^{(1)}=0\Leftrightarrow L_{1i}=L_{i1}=0$, which is equivalent to $[\delta_i,D]=0$ -- this condition can be used to invariantly characterize subfamilies of Kundt spacetimes.

\subsection{Kundt spacetimes of aligned Riemann type II}

\label{app_Kundt_II}

The above results hold for {\em any} Kundt spacetime. If one now restricts to the Kundt spacetimes of {\em aligned Riemann type II} (i.e., we assume $R_{010i}=0$ in addition to \eqref{RiemIb}),\footnote{This is equivalent to saying that $\bl$ is a multiply aligned null direction of both the Weyl and the Ricci tensors.} using \eqref{Kundt}, \eqref{Kundt_rot} the Bianchi identities (B3), (B5), (B12), (B1), (B6) and (B4) of \cite{Pravdaetal04} reduce to 
\beqn
 & & DR_{01ij}=0 , \label{B3} \\
 & & DR_{0i1j}=0  , \label{B5} \\
 & & DR_{ijkl}=0 , \label{B12} \\ 
 & & D R_{101i}-\delta_i R_{0101}=-R_{0101} L_{i1}-R_{01is} L_{s1}-R_{0i1s} L_{s1}  \label{B1} , \\
 & & D R_{1kij}+\delta_k R_{01ij}=R_{01ij}L_{k1}-2R_{0k1[i} L_{j]1}+R_{ksij} L_{s1}-2R_{01[i|s} \M{s}{|j]}{k} , \label{B6} \\
 & & D R_{1i1j}-\T R_{0j1i}-\delta_j R_{101i}=R_{0101} N_{ij}-R_{01is} N_{sj}+R_{0s1i} N_{sj}+R_{0j1s} \M{s}{i}{1}+R_{0s1i} \M{s}{j}{1}  \nonumber \\
 & &		\qquad\qquad\qquad\qquad\qquad\qquad\qquad		 {}+2R_{101i} L_{[1j]}+R_{1ijs} L_{s1}+R_{101s} \M{s}{i}{j} .   \label{B4}  
\eeqn

Note that here \eqref{11be} reduces to $DL_{1i}=0=DL_{i1}$, so that differentiation of \eqref{11a} gives
\be
	 D^2L_{11}=-DR_{0101} . \label{D11a} 
\ee

An important subset of Kundt spacetimes of Riemann type II is given by the {\em degenerate} Kundt metrics \cite{ColHerPel09a,Coleyetal09}, defined by
\begin{definition}[Degenerate Kundt metrics \cite{ColHerPel09a,Coleyetal09}] 
	A Kundt spacetime is ``degenerate'' if the Kundt null direction $\bl$ is also a multiple null direction of the Riemann tensor and of its covariant derivatives of arbitrary order (which are thus all of aligned type~II, or more special). 
	\label{def_deg}
\end{definition}
It is worth emphasizing that, in fact, the degenerate condition is automatically met at all orders once it is satisfied by the Riemann tensor and its {\em first} derivative (Theorem~4.2 and section~7 of \cite{Coleyetal09}). 
For degenerate Kundt spacetimes we have the following
\begin{proposition}[Conditions for degenerate Kundt metrics]
\label{prop_Kundt_deg}
A Kundt spacetime is degenerate iff it is of aligned Riemann type II and $\ell^a R_{,a}=0$ (using an affine parameter and a parallely transported frame, the latter condition is equivalent to any of the following: {$DR=0$, $DR_{01}=0$ {or (for $n>3$)} $DC_{0101}=0$\footnote{Recall that for $n=3$ the Weyl tensor vanishes identically and therefore the condition $DC_{0101}=0$ is trivial.}).} A Kundt spacetime for which the tracefree part of the Ricci tensor is of aligned type III is necessarily degenerate.
\end{proposition}
\begin{proof} 
It is a result of \cite{ColHerPel09a,Coleyetal09} that for Kundt spacetimes the degenerate condition is equivalent to the Riemann type II with $DR_{0101}=0$. The first part of the proposition thus simply follows from $DR_{0101}=0$ and the contracted Bianchi identities. The second part follows using Proposition~2 of \cite{OrtPraPra07} (implying the Weyl type II -- cf. also Proposition~7.1 of \cite{OrtPraPra13rev}) and, again, the contracted Bianchi identities (whose component of b.w. $+1$ gives $DR=0$).  
\end{proof} 
For Kundt spacetimes of Riemann type II, the degenerate condition is also equivalent to $D^2L_{11}=0$ (cf.~\eqref{D11a}). Note that the assumptions on the Ricci tensor in the second part of Proposition~\ref{prop_Kundt_deg} are of physical interest since they correspond to the case when the energy-momentum tensor is triply aligned with $\bl$. See Remark~\ref{rem_deg} for further comments.

An alternative covariant characterization of degenerate Kundt metrics was given in Proposition~6.1 of \cite{Coleyetal09} for $n=4$. This result in fact holds for any $n$ and we reproduce it here, along with the sketch of a proof different from the one of \cite{Coleyetal09} (i.e., not using the explicit form of the Kundt metric in adapted coordinates). After defining the symmetric 2-tensor
\be
	Q^{ab}\equiv R^{acbd}\pounds_{\bl}\pounds_{\bl}g_{cd} ,
\ee
we can state:
\begin{proposition}[Covariant characterization of degenerate Kundt metrics \cite{Coleyetal09}]
A Kundt spacetime is: 
	\begin{enumerate}[(i)]
		\item of aligned Riemann type II (or more special) iff $Q^{ab}Q_{ab}=0$ \label{cov_i}
		\item degenerate iff $Q^{ab}Q_{ab}=0$ and $\pounds_{\bl}\pounds_{\bl}\pounds_{\bl}g_{ab}=0$. \label{cov_ii}
	\end{enumerate}		
\end{proposition}
\begin{proof} 
In a Kundt spacetime, using an affine parameter from \eqref{dl} one obtains
\be
	\pounds_{\bl}g_{ab}= 2 L_{11} \ell_a \ell_b +(L_{1i}+L_{i1})(\ell_a \msub{i}{b}+\msub{i}{a} \ell_b) . \label{L_g}
\ee	
Taking the Lie derivative of this expression and using \eqref{dm} and \eqref{11be}, it is easy to see that $Q^{ab}Q_{ab}=0\Leftrightarrow R_{010i}=0$, which (recalling \eqref{RiemIb}) proves \eqref{cov_i}.

When $Q^{ab}Q_{ab}=0$, using a parallelly transported frame one easily finds that $\pounds_{\bl}\pounds_{\bl}\pounds_{\bl}g_{ab}=0\Leftrightarrow D^2L_{11}=0$, which (recalling \eqref{D11a}) proves \eqref{cov_ii}.

\end{proof}

\section{Proof of Theorem~\ref{theor}}
\setcounter{equation}{0}

\label{app_proof}

\subsection{Proof of ``\ref{cond2}. $\Rightarrow$ \ref{cond1}.''}

\label{sec_kundt->vsi}

{\subsubsection{Preliminaries}}

Before starting with the proof, let us make a few helpful observations on the strategy we shall adopt. First, by assumption~\ref{cond2a}, $\bF$ is VSI$_0$. If we are able to show that all its covariant derivatives are of aligned type III (or more special), then we are done with the proof.

Now, by assumption~\ref{cond2c}, $\bl$ is Kundt. Using an affine parameter and a frame parallely transported along $\bl$, this implies that the covariant derivatives of the frame vectors do not produce terms of higher b.w., cf. \eqref{dl}--\eqref{dm} (i.e., $\ell_{a;b}$ has only components of b.w. $-1$ or less, etc.). Together with assumption~\ref{cond2b} (which here can be written as the first of \eqref{DF_DR}), this immediately shows that $\nabla\bF$ is of aligned type III, as required   (cf. also Proposition~\ref{proposition_VSI1_2} in Appendix~\ref{app_VSI_1_2}). The problem is now to show that the same is true for covariant derivatives of $\bF$ of {\em arbitrary} order.

To this end, the {\em balanced-scalar approach} of \cite{Pravdaetal02,Coleyetal04vsi} will be useful. 
This approach can be applied to various tensors (or spinors). In the context of VSI tensors, the main idea is to show that (under proper assumptions) the covariant derivative of a tensor of type III is necessarily of aligned type III and thus, {\em by induction}, the tensor under consideration is VSI. Let us thus recall the relevant definition of \cite{Pravdaetal02,Coleyetal04vsi}:

\begin{definition}[Balanced scalars and tensors \cite{Pravdaetal02,Coleyetal04vsi}]
\label{def_balanced}
	In a frame parallely transported along an affinely parameterized geodesic null vector field $\bl$, a scalar $\eta$ of b.w. $b$ under a constant boost is a ``balanced scalar'' if $D^{-b}\eta=0$ for $b<0$ and $\eta = 0$ for $b\geq0$. A tensor whose components are all balanced scalars is a ``balanced tensor''.
\end{definition}

Note, in particular, that balanced tensors are of type III (or more special), multiply aligned with $\bl$. Restating the inductive method mentioned above in more technical terms, this will thus consist in showing that the covariant derivative of a balanced tensor is again a balanced tensor (Lemma~\ref{lemma_deriv} below, which will then apply to $\bF$).\footnote{
In the balanced-scalar approach, it is convenient to assign a b.w. to all the Newman-Penrose quantities, and this is why we consider only {\em constant} boosts in the Definition~\ref{def_balanced} (so that, e.g., $L_{11}$ has b.w. $-1$ and, if $\eta$ has b.w. $b$, then $D\eta$ and $\T\eta$ have, respectively, b.w. $(b+1)$ and $(b-1)$, etc. -- these quantities would {\em not} admit a b.w. under a general boost \cite{Durkeeetal10}). It is important to observe that there is no loss of generality here as far as our proof is concerned -- cf. also \cite{Coleyetal04vsi,HerPraPra14,Herviketal15}.}

\subsubsection{Proof}

We are thus ready to prove the direction ``\ref{cond2}. $\Rightarrow$ \ref{cond1}.'' of Theorem~\ref{theor}.

By assumption \ref{cond2c}, $\bl$ is degenerate Kundt (and thus geodesic). Eq.~\eqref{Kundt} is satisfied and, employing an affine parameter and a parallelly transported frame, also \eqref{Kundt_rot} and \eqref{RiemIb} hold, along with $R_{010i}=0$ and $DR_{0101}=0$.
Using the Ricci and Bianchi identities and the commutators summarized in appendix~\ref{app_D_Kundt}, one easily arrives at 
\beqn
 & & DL_{1i}=0 , \qquad DL_{i1}=0 , \qquad D\M{i}{j}{k}=0 , \label{Ricci_D_0} \\
 & & D^2N_{ij}=0 , \qquad D^2\M{i}{j}{1}=0 , \qquad D^2L_{11}=0 , \qquad  D^3N_{i1}=0 . \label{Ricci_D_-1-2}
\eeqn 
Together with the commutators \eqref{TD} and \eqref{dD}, this suffices to readily extend Lemma~4 of \cite{Coleyetal04vsi} 
(see \cite{Coleyetal04vsi,Pravdaetal02} for more technical details), i.e., 
\begin{lemma}[Balanced scalars in degenerate Kundt spacetimes]
\label{lemma_balanced}
In a degenerate Kundt spacetime, employing an affine parameter and a parallelly transported frame, if $\eta$ is a balanced scalar of b.w. $b$, then all the following scalars (ordered by b.w.) are also balanced: $D\eta$; $L_{1i}\eta$, $L_{i1}\eta$, $\M{i}{j}{k}\eta$, $\delta_i\eta$; $L_{11}\eta$, $N_{ij}\eta$, $\M{i}{j}{1}\eta$, $\T \eta$, $N_{i1}\eta$.	
\end{lemma}

Now, considering \eqref{dl}--\eqref{dm} and Lemma~\ref{lemma_balanced}, one can easily extend also Lemma~6 of \cite{Coleyetal04vsi} (cf. also \cite{Pravdaetal02} for a spinorial version of in four dimensions), i.e.,  
\begin{lemma}[Derivatives of balanced tensors in degenerate Kundt spacetimes]
\label{lemma_deriv}
 In a degenerate Kundt spacetime, the covariant derivative of a balanced tensor is again a balanced tensor.
\end{lemma}

Next, by assumptions \ref{cond2a} and \ref{cond2b}, we also have that \eqref{F_N} and the first of \eqref{DF_DR} hold. But these two equations precisely mean that {\em $\bF$ is a balanced tensor}. By Lemma~\ref{lemma_deriv}, the covariant derivatives of arbitrary order of $\bF$ are thus balanced tensors. In particular, they all possess only components of negative boost weight, which implies that $\bF$ is VSI, as we wanted to prove.

\subsection{Proof of ``\ref{cond1}. $\Rightarrow$ \ref{cond2}.''}

\label{sec_vsi->kundt}

\subsubsection{Preliminaries}

\label{subsubsec_prelim}

Let us start by proving two useful lemmas (of some interest in their own) about general tensors with vanishing invariants.\footnote{To avoid confusion, let us emphasize that in the present section~\ref{subsubsec_prelim}, and only here, $\bT$ can be any tensor and does not necessarily coincide with the energy-momentum tensor~\eqref{Energy momentum}.\label{foot_T}} {The first of these (Lemma~\ref{lemma_VSI})} follows directly from Theorem~2.1 {and Corollary~3.2} of \cite{Hervik11} (since $\bT$ and its covariant derivatives up to order $I$ are clearly not characterized by their invariants). We nevertheless provide an independent simple proof.

\begin{lemma}[{Alignment of VSI tensors}]
\label{lemma_VSI}
  If a tensor field $\bT$ is VSI$_I$, $\bT$ and its covariant derivatives up to order $I$ are of {\em aligned} type III (or more special).
\end{lemma}
\begin{proof} 
The case $I=0$ is contained in the algebraic VSI theorem (Theorem~\ref{alg_theor}), so we need to discuss only the case $I\ge 1$. Additionally, the lemma is trivially true in the case $\nabla\bT=0$, so that we can assume hereafter $\nabla\bT\neq0$. $\bT$ being VSI$_I$ {implies} that $\bT$ and its covariant derivatives up to order $I$ are all VSI$_0$ (Definition~\ref{def_VSI}). That all these tensors are of type III thus follows immediately from \cite{Hervik11}. It remains to be proven that they are all aligned (i.e., for all of them the same null vector $\bl$ defines a multiple null direction such that all the non-negative b.w. components vanish).

Since	$\bT$ is of type III, it admits a {\em unique} multiply aligned null direction such that all the non-negative b.w. components vanish. Furthermore, such tensor cannot admit a distinct null direction with respect to which {all} positive b.w. components vanish.  Let us work in a null frame such that $\bl$ is parallel to this {unique} null direction. In this frame, $\bT$ possesses only components of negative b.w.. Now, let us consider two possibile cases separately. $i)$ First, if $\bT$ has only components of b.w. $-2$ or less, then the components of $\nabla\bT$ have b.w. 0 or less (since the covariant derivative of a tensor can, at most, raise the b.w. of $+2$), i.e., also $\nabla\bT$ is multiply aligned to $\bl$. But $\nabla\bT$ must be of type III (as noticed above), therefore the only possibility is that the components of $\nabla\bT$ have, in fact, b.w. $-1$ (or less) in the frame we are using, so $\nabla\bT$ is of {\em aligned} type III. $ii)$ On the other hand, if $\bT$ has some components of b.w. $-1$, we cannot use the same argument, since some components of $\nabla\bT$ will have b.w. $+1$, in general. Let us thus assume this is indeed the case (if not, i.e., if components of $\nabla\bT$ have only b.w. $0$ or less, then we can proceed as in case $i)$) and let us consider, instead, the tensor product $\bT_\times\equiv\bT\times\nabla\bT$. Obviously, the components of $\bT_\times$ cannot have b.w. greater then $0$ (cf., e.g., Proposition~A.11 of \cite{HerOrtWyl13}). But since $\bT$ is VSI$_1$, then $\bT_\times$ must be of type III, which thus implies that the components of {$\bT_\times$} can only have b.w. $-1$ or smaller in a frame adapted to $\bl$. However, it is not difficult to see that this cannot be true if $\nabla\bT$ possesses some components of b.w. $+1$ (as we assumed), thus leading to a contradiction. In other words, if $\bT$ is VSI$_1$ then $\nabla\bT$ can only have components of b.w. $0$ or less; but then, in fact, these can be only of b.w. $-1$ or less (since $\nabla\bT$ must be of type III, similarly as in point $i)$), so that, again, $\nabla\bT$ is of {\em aligned} type III.

Combining $i)$  and $ii)$ we have proven the lemma for the case $I=1$. Clearly the same argument extends to any higher $I$, i.e., the proof is complete.

\end{proof}

 In turn, this can be used to prove the following result for rank-2 tensors.

\begin{lemma}[$T_{ab}$ {VSI$_2$} implies Kundt]
\label{lemma_VSI_2}
  If a 2-tensor field $T_{ab}$ is VSI$_2$ then : (a) $T_{ab}$, $T_{ab;c}$ and  $T_{ab;cd}$ are of aligned type III (or more special); (b) the corresponding multiple null direction $\bl$ is necessarily Kundt.
\end{lemma}

\begin{proof}

Point (a) follows immediately from Lemma~\ref{lemma_VSI}. It will be used in the following.

Now, thanks to (a), in an adapted null frame we can write
\be
 T_{ab}=T_{1i}\ell_a m^{(i)}_{\,b}+T_{i1}m^{(i)}_{\,a}\ell_b+T_{11}\ell_a\ell_b . 
 \label{T_frame}
\ee
It is convenient to first prove (b) in the special case $T_{1i}=0=T_{i1}$, i.e., $T_{ab}=T_{11}\ell_a\ell_b$ (of course with $T_{11}\neq0$). Let us define a compact notation for the covariant derivatives 
\be
 T^{(I)}_{abc_1\ldots c_I}\equiv T_{ab;c_1\ldots c_I} \qquad  (I=1,2,3,\dots) .
 \label{T^I}
\ee
For our purposes, it will now suffice to require that certain components of b.w. 0 of $\bT^{(1)}$ and $\bT^{(2)}$ vanish (in view of (a)). First, requiring $T^{(1)}_{i10}=0$  one obtains (using also the first of \eqref{Ricci_rot})
\be
	L_{i0}=0 , 
	\label{geod}
\ee
i.e., $\bl$ must be geodesic (note that the VSI$_1$ property of $\bT$ suffices to prove this). Next, the condition $T^{(2)}_{01ij}=0$ is equivalent to $L_{ki}L_{kj}=0$, which, by tracing, leads to 
\be
 L_{ij}=0 ,
 \label{kundt}
\ee
i.e. (together with \eqref{geod}), $\bl$ must be a Kundt null direction, which thus proves (b) in the case $T_{ab}=T_{11}\ell_a\ell_b$.

Finally, for a generic $T_{ab}$ (eq.~\eqref{T_frame}) we can apply the same argument to the tensor $\tilde T_{ab}\equiv T_{ac}T_b^{\ c}=(T_{1i}T_{1i})\ell_a\ell_b$  if $T_{1i}\neq0$ (or to $\tilde T_{ab}\equiv T_{ca}T^{c}_{\ b}=(T_{i1}T_{i1})\ell_a\ell_b$ if $T_{i1}\neq0$), so the proof of (b) is now complete (where we used the fact that if $T_{ab}$ is VSI$_2$ then $\tilde T_{ab}$ must obviously also be VSI$_2$).

\end{proof}

\subsubsection{Proof}

Let us now prove the direction ``\ref{cond1}. $\Rightarrow$ \ref{cond2}.'' of Theorem~\ref{theor}. We assume that $\bF$ is VSI. It is, in particular, VSI$_0$ and thus, by {Corollary}~\ref{lemma_VSI0}, $\bF$ must be of type N, i.e., condition~\ref{cond2a} is proven. In an adapted frame this means that \eqref{F_N} holds. It remains to prove that conditions~\ref{cond2b} and \ref{cond2c} are also satisfied, i.e., we need to show (recall \eqref{kundt1} and \eqref{DF_DR})
\begin{enumerate}[(i)]
	\item $L_{i0}=0$, $L_{ij}=0$	\label{cond_Kundt}	\label{thes_kundt}
	\item $R_{010i}=0$, $DR_{0101}=0$ (in a parallelly transported frame) \label{thes_deg}
	\item $DF_{1i_1\ldots i_{p-1}}=0$ (in a parallelly transported frame). \label{thes_DF} 
\end{enumerate}

It is convenient to define the following 2-tensor $\bT$\footnote{For a VSI$_0$ field $\bF$, $\bT$ equals the associated energy-momentum tensor \eqref{Energy momentum}, since $F^2=0$.}
\be
  T_{ab} = {\frac{\b}{8\pi}} F_{ac_1\ldots c_{p-1}} {F_b}^{c_1\ldots c_{p-1}} .
	\label{Energy_momentum}
\ee
Since $\bF$ is VSI, then {\em $\bT$ must also be VSI}, thus condition~\eqref{thes_kundt} follows immediately from Lemma~\ref{lemma_VSI_2} applied to $\bT$, so that $\bl$ is Kundt. (Alternatively, condition~\eqref{thes_kundt} can also be proven using Proposition~\ref{proposition_VSI1_2} in Appendix~\ref{app_VSI_1_2} -- note also that so far we used only the assumption that $\bF$ is VSI$_2$.)

For the next steps, it is useful to observe that, by \eqref{F_N}, $\bT$ is also of type N w.r.t. $\bl$, i.e.,
\be
  T_{ab} = {\frac{\b}{8\pi}} \F^2\ell_a \ell_b ,
	\label{T_N}
\ee
where 
\be
	\F^2\equiv F_{1i_1\ldots i_{p-1}}F_{1i_1\ldots i_{p-1}}\ge 0 ,
	\label{F2}
\ee	
and has b.w. $-2$. Since the b.w. of $\bF$ and $\nabla^{(I)}\bF$ is always $-1$ or less (Lemma~\ref{lemma_VSI}), also the b.w. of $\bT^{(I)}$ (defined in \eqref{T^I}) must always be $-2$ (or less; cf. also Proposition~A.11 of \cite{HerOrtWyl13}). This condition will be used below.

From now on, we employ an affine parameter along the Kundt null vector $\bl$ and a frame parallely transported along it, so that~\eqref{Kundt_rot} holds. As observed in appendix~\ref{app_D_Kundt}, this implies $R_{0i0j}=0=R_{0ijk}=0$ (eq.~\eqref{RiemIb}).

Next, requiring $\nabla\bF$ to be of type III {(namely, $(\nabla\bF)_{01i_1\ldots i_{p-1}}=0$)} is now equivalent to condition~\eqref{thes_DF}, which is thus also proven (again, this alternatively follows from Proposition~\ref{proposition_VSI1_2}).

Condition~\eqref{thes_DF} in turns implies $D(\F^2)=0$, so that by a boost we can set $\F^2=1$ in \eqref{T_N} (while preserving the affine parametrization of $\bl$ and the parallel transport of the frame), i.e., from now on $T_{ab} ={\frac{\b}{8\pi}} \ell_a \ell_b$. 

Now, imposing $T^{(2)}_{11i0}=0$ and $T^{(2)}_{i110}=0$ (these components have b.w. $-1$ and thus must vanish, as observed above) gives, respectively,
\be
 DL_{1i}=0 , \qquad DL_{i1}=0 .
 \label{DL1i}
\ee
By \eqref{11be} this in turn implies
\be
  R_{010i}=0 ,
\ee
so that the Riemann type is II or more special (recall that we already obtained $R_{0i0j}=0=R_{0ijk}$ above).

Finally, requiring $T^{(3)}_{11100}=0$ (b.w. $-1$) leads to
\be
 D^2L_{11}=0 , 
\ee
and thus, by~\eqref{11a} (with \eqref{DL1i}), 
\be
 DR_{0101}=0 ,
\label{DR0101}
\ee
which completes the proof of condition~\eqref{thes_deg}, and the proof is now complete.

Note that, in fact, we have used only up to the third derivatives of $\bF$ in the argument above (cf. Remark~\ref{rem_VSI}).

\section{VSI$_1$ and VSI$_2$ $p$-forms}
\setcounter{equation}{0}

\label{app_VSI_1_2}

A bridge between Corollary~\ref{lemma_VSI0} and Theorem~\ref{theor} is provided by the following result (see also Remark~\ref{rem_VSI}).

\begin{proposition}[VSI$_1$ and VSI$_2$ $p$-forms]
 \label{proposition_VSI1_2}
	A $p$-form $\bF$ is VSI$_1$ iff it is is of type N, {$\pounds_{\bl}\bF=0$}, $\bl$ is Kundt. It is VSI$_2$ iff it is is of type N, {$\pounds_{\bl}\bF=0$}, $\bl$ is Kundt and (at least) doubly aligned with the Riemann tensor.
\end{proposition}

In particular, this means that {\em for solutions of the Einstein-Maxwell theory, we have that VSI$_1\Rightarrow$VSI}, since condition~\ref{cond2c} of Theorem~\ref{theor} is automatically satisfied (cf. Remark~\ref{rem_deg}). {Similarly, VSI$_1\Rightarrow$VSI also for a $\bF$ in the classes of spacetimes mentioned in Remark~\ref{rem_deg} which are necessarily degenerate Kundt.}

\begin{proof} 

Let us first prove the ``only if'' part. We assume that $\bF$ is VSI$_1$. Then $\bT$ must have the same property. In particular, in an adapted frame, we have~\eqref{T_N}. Further, requiring $T^{(1)}_{i10}=0$  (recall definition~\eqref{T^I}) we obtain $L_{i0}=0$ (i.e., $\bl$ is geodesic). From now one we can thus employ an affine parameter and a frame parallelly transported along $\bl$. The condition $(\nabla\bF)_{01i_1\ldots i_{p-1}}=0$ gives $DF_{1i_1\ldots i_{p-1}}=0$. From $(\nabla\bF)_{i_1\ldots i_{p+1}}=0$ and $(\nabla\bF)_{j10i_1\ldots i_{p-2}}=0$ we obtain, respectively,
\beqn
  & & F_{1[i_1\ldots i_{p-1}}L_{i_p]j}=0 , \label{VSI1_a} \\
	& & F_{1ji_1\ldots i_{p-2}}L_{jk}=0 . \label{VSI1_b}
\eeqn
Contracting \eqref{VSI1_a} with $L_{i_p k}$ and using \eqref{VSI1_b} leads to $F_{1i_1\ldots i_{p-1}}L_{jk}L_{jl}=0$. Further contraction with $\delta_{kl}$ gives $L_{jk}L_{jk}=0$ and thus $L_{jk}=0$, i.e., $\bl$ is Kundt. This is all we needed to prove as for the VSI$_1$ statement. 

If, additionally, $\bF$ is VSI$_2$, then from $T^{(2)}_{11i0}=0$ and $T^{(2)}_{i110}=0$ we obtain $DL_{1i}=0=DL_{i1}$. By~\eqref{11be} this implies $R_{010i}=0$, so that (recall also \eqref{RiemIb}) $\bl$ is doubly aligned with the Riemann tensor, as we wanted to prove.

The ``if'' part of the proposition can be proven similarly by reversing the above steps (essentially showing that under the  conditions given in Proposition~\ref{proposition_VSI1_2} the first covariant derivatives (and for VSI$_2$ also the second covariant derivatives) of $\bF$ are b.w. negative.

\end{proof}


\providecommand{\href}[2]{#2}\begingroup\raggedright\endgroup

\end{document}